\documentclass[letterpaper]{article}

\ifdefined\aaaianonymous
    \usepackage[submission]{aaai2026}
\else
    \usepackage{aaai2026}
\fi

\usepackage{times}
\usepackage{helvet}
\usepackage{courier}
\usepackage[hyphens]{url}
\usepackage{graphicx}
\usepackage{natbib}
\usepackage{caption}
\usepackage{amsmath}
\usepackage{amssymb}
\usepackage{amsthm}

\usepackage{mathtools}
\usepackage{booktabs}
\usepackage{array}
\usepackage{placeins}
\usepackage{float}
\usepackage{multirow}
\usepackage{enumitem}
\newtheorem{theorem}{Theorem}
\newtheorem{proposition}{Proposition}

\title{Same Weights, Different Robot:\\ A Deployment Safety View of VLA Policies}

\author{Jianwei Tai}
\affiliations{School of Internet, Anhui University\\
24012@ahu.edu.cn}

\begin{document}
\maketitle
\emergencystretch=3em
\setlength{\hfuzz}{2pt}

\begin{abstract}
Vision-language-action (VLA) policies are often treated as checkpoint-defined objects: if the weights, prompt, and benchmark suite match, the deployment is assumed to be the same policy.  Robot execution breaks this assumption because the same normalized model output can become a different physical action after action unnormalization and controller conventions are applied.  This creates a deployment-safety gap: safety review can certify the checkpoint while missing the executable robot policy that reaches the controller.  We formalize this gap as an executable policy specification problem: a VLA policy includes the learned model, action representation, metadata-selected unnormalizer, and controller-facing conventions.  Under this view, identical checkpoints can be executable-inequivalent.  For quantile-style action normalization, we derive a closed-form metadata mismatch transform and an ExecSpec certificate that measures action-space semantic drift without model inference or rollout.  On LIBERO-Goal replay, substituting a plausible sibling metadata key yields mean drift 0.199 over six non-gripper action dimensions and reduces success from 28/28 to 2/28 under full substitution.  On LIBERO-Spatial replay, the same substituted key reduces success from 26/26 to 0/26.  The same full-substitution protocol gives 0/28 success for all four Object substitutions and 0/23 or 1/23 success on Long.  Identity-key, replay-validity, no-op filtering, raw-vs-correct replay, mask/gripper, synthetic upper-bound, and OpenVLA-style unnormalizer interface checks rule out several simpler explanations.  These results do not certify closed-loop or hardware safety.  They support a narrower deployment-safety view: action-space metadata is part of the executable policy and should be checked before rollout.
\end{abstract}

\section{Introduction}
\label{sec:introduction}

Vision-language-action (VLA) policies are moving robot control toward a checkpoint-centric deployment model.  A policy is downloaded, paired with a processor, given an instruction, and expected to emit actions in a benchmark or robot environment.  This workflow makes model weights the visible unit of reproducibility and safety review.  A deployment is often cleared for evaluation once the checkpoint hash, prompt template, and benchmark suite match a tested artifact.

For robots, that assumption is a deployment-safety risk, not merely a bookkeeping problem.  The model output is not the physical action.  Between the normalized action representation and the controller lies an executable specification: image preprocessing, action token decoding, action unnormalization, gripper convention, control frequency, and sometimes chunk execution.  These components define the controller-facing semantics of the action space.  When they are omitted from the policy object, the same weights and the same normalized outputs can instantiate a different executable robot policy: the controller receives different physical actions even though the checkpoint appears unchanged.

This paper studies one concrete instance of that safety boundary: action normalization metadata.  In many VLA and imitation-learning pipelines, the model emits normalized action coordinates, and a dataset-specific key selects quantile statistics used to map those coordinates back to physical actions.  Such metadata is easy to treat as a configuration detail.  It is also easy to copy from a sibling dataset, checkpoint, or suite.  The resulting error is a silent interface failure: generated tokens and normalized actions remain unchanged, yet the deployment executes a different robot policy.  This is precisely the kind of mismatch that a pre-rollout safety audit should detect before policy quality or hardware behavior is evaluated.

The point is not simply that an incorrect statistic can cause a bad run.  The paper's object is policy identity: any field that changes the controller-facing action law is part of the executable policy, even when it is stored outside the neural checkpoint.  Under this view, metadata mismatch is a pre-deployment executable non-equivalence, not merely a post-hoc configuration bug.  It can therefore be checked as a specification condition before a model is rolled out.

We argue that action normalization metadata is a safety-relevant component of the executable policy.  This is not a claim that every metadata mismatch is catastrophic, nor that normalization is the only hidden deployment hazard.  The claim is sharper: if a deployment changes the action unnormalizer while holding the model and normalized action sequence fixed, it changes the executable policy.  For quantile normalizers this change is a closed-form affine transform in physical action space.  It can be measured before model inference, simulation rollout, or robot execution.

The safety issue is that this transform can cross task margins.  Manipulation safety is often margin-based: approach direction, contact velocity, rotation, and gripper timing must remain within a local success basin, especially near contact.  A displacement that is small relative to the full action range can still be large relative to a grasp or contact margin.  A safety review that ignores action-space metadata can therefore approve a checkpoint while missing the physical action law that is actually deployed.

We formalize this audit boundary through executable policy equivalence.  A VLA policy is represented as an executable tuple
\begin{equation}
(f_\theta,\mathcal I,\mathcal Z,\mathcal U_k,\mathcal C),
\end{equation}
where \(\mathcal U_k\) is the metadata-selected unnormalizer and \(\mathcal C\) contains controller conventions.  Two deployments are executable-equivalent only if they induce the same physical action law.  This makes checkpoint equality insufficient: identical weights and prompts do not imply identical robot actions.

We then derive an ExecSpec certificate for metadata mismatch.  For quantile-style unnormalization, replacing key \(k_c\) with key \(k_w\) induces a deterministic per-dimension affine displacement
\begin{equation}
\begin{aligned}
\Delta_i(z;k_c,k_w)
&=(q^{(w)}_{01,i}-q^{(c)}_{01,i}) \\
&\quad + z_i\Big[(q^{(w)}_{99,i}-q^{(w)}_{01,i})
        -(q^{(c)}_{99,i}-q^{(c)}_{01,i})\Big]
\end{aligned}
\end{equation}
on masked action dimensions.  The certificate reports mean displacement, tail displacement, axis attribution, mask mismatch, and gripper displacement on a calibration set.  It requires no model inference and serves as a pre-rollout interface screen: a large certificate says the deployment is not the intended executable policy before the robot moves.

Empirically, plausible metadata substitutions are sufficient to flip replayed robot task success.  On LIBERO-Goal, using \textsc{goal-B} as the intended key, substituting \textsc{long-v2} yields mean drift 0.199 over six non-gripper action dimensions, p95 drift 0.275, and 43.5\% of calibration frames above drift 0.20.  On replay-valid demonstrations, success under \textsc{long-v2} drops from 28/28 at the intended metadata to 2/28 at full substitution, with a dose-response curve of 28/28, 17/28, 7/28, 0/28, and 2/28 as interpolation moves from the correct key to the wrong key.  On LIBERO-Spatial, using \textsc{spatial-v2} as the intended key, \textsc{long-v2} yields 0/26 success at full substitution.  On LIBERO-Object, all four plausible sibling substitutions reduce replay success from 28/28 to 0/28 on replay-valid demonstrations.  On LIBERO-Long, \textsc{goal-B} and \textsc{spatial-v2} reduce success from 23/23 to 0/23, while \textsc{mixed-all} reduces it to 1/23.  Identity-key, no-op filtering, raw-vs-correct replay, mask, gripper, synthetic upper-bound, and OpenVLA-style interface checks rule out several simpler explanations.  These replay witnesses support the deployment-safety concern without claiming to certify closed-loop or hardware safety.

The contribution is a deployment-safety view of executable VLA policies:
\begin{enumerate}
    \item We define an executable-policy identity criterion for VLA deployment and show that checkpoint equality is insufficient for physical policy equality.
    \item We derive a closed-form metadata mismatch transform and an exact pre-rollout certificate for one controller-facing semantic field: quantile-style action unnormalization metadata.
    \item We provide LIBERO-Goal, LIBERO-Spatial, LIBERO-Object, and LIBERO-Long replay witnesses that this identity violation is physically consequential, with identity-key, no-op filtering, raw-vs-correct replay, mask, gripper, synthetic upper-bound, and interface-sanity controls ruling out simpler explanations.
\end{enumerate}

The paper is not a general claim that all VLA deployments fail under metadata mismatch.  It identifies a hidden action-space semantics layer that must be specified for a robot policy to be reproducible and auditable.  The broader deployment lesson is that artifacts outside the checkpoint can be load-bearing parts of the executable policy.

\section{Related Work}
\label{sec:related}

\paragraph{Vision-language-action policies.}
VLA models extend vision-language pretraining to robot control by mapping images and language instructions to actions \cite{openvla,rt2,octo,cogact}.  Released checkpoints have become a natural unit of comparison for tasks, objects, and instructions.  Large robot data mixtures further make action-interface assumptions harder to inspect manually \cite{openxembodiment,droid,robonet}.  Evaluation protocols usually specify the benchmark suite, model checkpoint, prompt style, and sometimes image preprocessing.  Recent VLA infrastructure exposes OpenVLA-style and LeRobot-style unnormalization keys or normalization assets as deployment parameters \cite{starvla,cogact,openpi_norm_stats,lerobot}, which reinforces that these fields are operationally load-bearing.  We focus on the formal consequence of selecting the wrong key: the executable physical action law changes, and that change can be certified before rollout.

\paragraph{Robot imitation learning and action normalization.}
Imitation-learning policies commonly normalize continuous actions using dataset statistics, quantiles, or standardization parameters \cite{bc,act,diffusion_policy,robomimic}.  Action representation choices such as absolute, relative, and delta control are deployment-facing conventions rather than training details alone \cite{umi,lerobot}.  Normalization stabilizes training and makes heterogeneous action dimensions easier to model.  In deployed policies, however, the inverse map becomes part of the controller-facing policy.  A normalization key that is harmless during data loading can become safety-critical when it selects the physical action semantics of a checkpoint.  Prior work typically treats normalization as an implementation convention.  We formalize the inverse normalization map as an executable component and show that substituting plausible sibling statistics can flip replayed task success.

\paragraph{Safety and robustness in robot learning.}
Robot-learning safety often studies distribution shift, compounding errors, unsafe exploration, adversarial perception, constraint satisfaction, or recovery from out-of-distribution states \cite{robot_safety_survey,safe_il,adversarial_robotics}.  Those hazards usually involve a changed environment, perturbed observation, policy error, or planner/controller failure.  Metadata mismatch is different: the observation, normalized action sequence, and initial state can all remain fixed.  The failure occurs because the same internal action is decoded into a different physical action.  This places the problem closer to deployment specification than to standard robustness.

\paragraph{Reproducibility and model cards.}
ML reproducibility tools emphasize dataset versions, random seeds, code, checkpoints, and evaluation scripts \cite{model_cards,datasheets,reproducibility}.  Robot-learning simulators and offline-imitation frameworks have also pushed reproducible benchmark setup and dataset handling \cite{robosuite,robomimic}.  Recent VLA benchmarks and harnesses extend this direction for robot policies by standardizing evaluation environments, model servers, and real-world benchmark setups \cite{vla_eval,vla_replica}.  Robot policies need a stronger notion of reproducibility because the executable artifact includes controller semantics.  A model card can identify the correct checkpoint while omitting the unnormalization key or gripper convention that determines the action reaching the robot.  Our executable-policy view complements documentation work by identifying a concrete missing field: action-space metadata must be versioned, hashed, and checked as part of the policy manifest.

\paragraph{Configuration bugs versus executable semantics.}
A natural objection is that wrong metadata is merely a configuration bug.  Configuration errors are usually treated as accidental deviations from a fixed policy.  ExecSpec makes the opposite point: if a field changes the controller-facing action law, then it is part of the policy identity.  The operational distinction is that the mismatch can be detected before failure occurs: changing the metadata key changes executable equivalence, yields a closed-form action transform, and produces a static pre-rollout certificate.  The empirical result then tests whether this formally specified transform is physically consequential.  The answer is yes across four LIBERO suite families under replay.  The contribution is therefore not that configuration errors exist, but that controller-facing action metadata defines policy identity and can be checked before deployment.

\section{Executable Policy Specification}
\label{sec:execspec}

A deployed robot policy is often identified with a checkpoint.  For vision-language-action (VLA) models this identification is incomplete.  The checkpoint maps an observation and a language instruction into an internal action representation, but the robot executes a physical control vector only after a sequence of preprocessing, decoding, unnormalization, and control-convention choices.  A safety audit that checks only the model weights can therefore certify the wrong object: two deployments can share identical weights and prompts while issuing different physical actions.

We formalize this distinction by treating the executable policy as the object that reaches the robot controller, not merely the neural network that emits normalized action coordinates.  Let an observation history be denoted by
\(
  h_t=(o_{\le t}, x_{\le t})
\),
where \(o_t\) includes the image and language-conditioned observation and \(x_t\) denotes any proprioceptive or controller state exposed to the policy.  An executable VLA policy is a tuple
\begin{equation}
\Pi_{\mathrm{exec}}
  = (f_\theta,\mathcal I,\mathcal Z,\mathcal U_k,\mathcal C).
\end{equation}
Here \(f_\theta\) is the learned model, \(\mathcal I\) is the image and language processor, \(\mathcal Z\) is the normalized action representation emitted or decoded by the model, \(\mathcal U_k\) is the action unnormalizer selected by metadata key \(k\), and \(\mathcal C\) collects controller conventions such as gripper semantics, action dimensionality, control rate, and chunk execution.  The physical action sent to the environment is
\begin{equation}
  a_t = \mathcal C\!\left(\mathcal U_k(z_t)\right),
  \qquad
  z_t \sim f_\theta(\mathcal I(h_t)).
\end{equation}
The notation allows deterministic policies by replacing the sampling relation with \(z_t=f_\theta(\mathcal I(h_t))\).

\paragraph{Executable equivalence.}
Two deployments are executable-equivalent if they induce the same physical action law under every admissible history.  Formally,
\begin{equation}
\begin{aligned}
\Pi_1 \equiv_{\mathrm{exec}} \Pi_2
\quad\Longleftrightarrow\quad
&\mathcal L_{\Pi_1}(a_t\mid h_t) \\
&=\mathcal L_{\Pi_2}(a_t\mid h_t)
\quad \text{for all } h_t,t.
\end{aligned}
\label{eq:exec-equivalence}
\end{equation}
where \(\mathcal L_\Pi(a_t\mid h_t)\) denotes the distribution of physical actions issued by executable policy \(\Pi\).  This definition is intentionally stricter than checkpoint equality.  It is the physical action law, not the token distribution or normalized action vector, that determines whether a deployment can collide, miss a grasp, close a gripper at the wrong time, or fail to recover from a near-miss.

The safety-relevant failure mode studied in this paper is executable non-equivalence under identical model outputs.  Consider two deployments
\begin{equation}
\begin{aligned}
\Pi_c&=(f_\theta,\mathcal I,\mathcal Z,\mathcal U_{k_c},\mathcal C),\\
\Pi_w&=(f_\theta,\mathcal I,\mathcal Z,\mathcal U_{k_w},\mathcal C),
\end{aligned}
\end{equation}
which differ only in the metadata key used to select the unnormalizer.  The neural model, prompt, image processor, controller convention, initial state, and normalized action sequence are all held fixed.  If \(\mathcal U_{k_c}\ne \mathcal U_{k_w}\), then the two deployments can be non-equivalent even though their checkpoints and normalized outputs are identical.  This is a silent deployment hazard: the mismatch is not visible from generated tokens, normalized actions, or checkpoint hashes, but it changes the action that reaches the robot.

\begin{proposition}[Checkpoint equality is not executable equality]
\label{prop:checkpoint-not-exec}
There exist two executable VLA policies \(\Pi_c\) and \(\Pi_w\) with identical \(f_\theta\), \(\mathcal I\), \(\mathcal Z\), and \(\mathcal C\), and identical normalized action outputs for every history, such that \(\Pi_c\not\equiv_{\mathrm{exec}}\Pi_w\).
\end{proposition}

\begin{proof}
Choose any history \(h_t\) for which the model emits a normalized action \(z_t\).  Let the two deployments differ only by unnormalizers \(\mathcal U_{k_c}\) and \(\mathcal U_{k_w}\).  If \(\mathcal C(\mathcal U_{k_c}(z_t))\ne \mathcal C(\mathcal U_{k_w}(z_t))\), then the conditional physical action laws differ at \(h_t\).  By Eq.~\eqref{eq:exec-equivalence}, the deployments are not executable-equivalent, despite sharing the same checkpoint and normalized output.  Such a pair exists whenever the metadata keys induce distinct physical unnormalization maps on any normalized action in the support of \(f_\theta\).
\end{proof}

This proposition is simple, but it changes the audit boundary.  A deployment manifest that records only model weights, tokenizer, and prompt template is insufficient to specify the policy that the robot executes.  The action-space semantics must be versioned and checked as part of the policy object.  The next section instantiates this principle for the quantile-style action normalizers used by common VLA and imitation-learning pipelines, where metadata mismatch yields a closed-form affine transform in physical action space.

\section{Metadata Mismatch Certificates}
\label{sec:certificate}

We derive a certificate for the executable non-equivalence induced by action-normalization metadata.  The certificate is model-independent: it does not require running a VLA, sampling an action token, or rolling out a simulator.  It measures a narrower safety boundary.  When a deployment uses metadata key \(k_w\) where key \(k_c\) was intended, the certificate quantifies how far the physical action semantics move on a calibration set of normalized actions.

\subsection{Quantile unnormalization as an executable map}

Many robot imitation and VLA pipelines train the model to emit normalized action coordinates.  A metadata key \(k\) specifies per-dimension lower and upper quantiles \(q^{(k)}_{01},q^{(k)}_{99}\in\mathbb R^d\) and a mask \(m^{(k)}\in\{0,1\}^d\) indicating which dimensions are normalized.  Write \(\Delta q^{(k)}_i=q^{(k)}_{99,i}-q^{(k)}_{01,i}\).  For a normalized action \(z\in\mathbb R^d\), the unnormalizer is
\begin{equation}
\mathcal U_k(z)_i=
\begin{cases}
q^{(k)}_{01,i}+z_i\Delta q^{(k)}_i, & m^{(k)}_i=1,\\
z_i, & m^{(k)}_i=0.
\end{cases}
\label{eq:quantile-unnorm}
\end{equation}
The unmasked case is common for gripper dimensions that are already represented in a controller convention.  The analysis below allows mask mismatch, but our main empirical substitutions have identical masks and zero gripper-dimension drift; the failures are driven by non-gripper action dimensions.

Let \(k_c\) be the intended metadata key and \(k_w\) a substituted key.  The physical action displacement before controller post-processing is
\begin{equation}
\Delta(z;k_c,k_w)=\mathcal U_{k_w}(z)-\mathcal U_{k_c}(z).
\end{equation}
For a shared normalized action \(z\), this displacement is the entire effect of the metadata substitution.  All model-side behavior is held fixed.

\begin{theorem}[Closed-form metadata displacement]
\label{thm:metadata-displacement}
Assume two metadata keys \(k_c,k_w\) use the same mask \(m\).  For every masked dimension \(i\) with \(m_i=1\), substituting \(k_w\) for \(k_c\) changes the executed pre-controller action by
\begin{equation}
\begin{aligned}
\Delta_i(z;k_c,k_w)
&=(q^{(w)}_{01,i}-q^{(c)}_{01,i}) \\
&\quad + z_i(\Delta q^{(w)}_i-\Delta q^{(c)}_i).
\end{aligned}
\label{eq:closed-form-delta}
\end{equation}
For every unmasked dimension \(i\) with \(m_i=0\), \(\Delta_i(z;k_c,k_w)=0\).  Therefore the metadata substitution induces a deterministic per-dimension affine transform of the normalized action, independent of the model that produced \(z\).
\end{theorem}

\begin{proof}
For masked dimensions, expand Eq.~\eqref{eq:quantile-unnorm} under each key:
\begin{align}
\mathcal U_{k_w}(z)_i-
\mathcal U_{k_c}(z)_i
&=q^{(w)}_{01,i}+z_i\Delta q^{(w)}_i
 -q^{(c)}_{01,i}-z_i\Delta q^{(c)}_i \\
&=(q^{(w)}_{01,i}-q^{(c)}_{01,i})
+z_i(\Delta q^{(w)}_i-\Delta q^{(c)}_i).
\end{align}
For unmasked dimensions, both keys return \(z_i\) when the mask is shared, so the displacement is zero.  The expression depends only on metadata and \(z\), not on \(f_\theta\).
\end{proof}

\begin{proposition}[Metadata mismatch implies executable non-equivalence on support]
\label{prop:metadata-non-equivalence}
Consider two deployments that share \(f_\theta\), \(\mathcal I\), \(\mathcal Z\), and controller convention \(\mathcal C\), but use metadata keys \(k_c\) and \(k_w\).  If there exists a normalized action \(z\) in the support of the deployment such that
\begin{equation}
\mathcal C(\mathcal U_{k_c}(z))\ne \mathcal C(\mathcal U_{k_w}(z)),
\label{eq:controller-non-equivalence}
\end{equation}
then the two deployments are not executable-equivalent.  In particular, when \(\mathcal C\) is injective on the affected action coordinates, any nonzero metadata displacement on support is sufficient for executable non-equivalence.
\end{proposition}

\begin{proof}
The two deployments assign different controller-facing actions to the same history that produces \(z\).  Their conditional physical action laws therefore differ at that history, which violates Eq.~\eqref{eq:exec-equivalence}.  If \(\mathcal C\) is injective on the affected coordinates, \(\mathcal U_{k_c}(z)\ne\mathcal U_{k_w}(z)\) implies Eq.~\eqref{eq:controller-non-equivalence}.
\end{proof}

Theorem~\ref{thm:metadata-displacement} is not a performance bound.  It is a specification result: once \(z\) and the two metadata records are known, the action-space semantic shift is known exactly.  Proposition~\ref{prop:metadata-non-equivalence} states the policy-identity consequence.  The dangerous transformation is not inferred post hoc from failures; it is present before the robot moves.

\subsection{Certificates on calibration actions}

Let \(Z=\{z_j\}_{j=1}^n\) be a calibration set of normalized actions.  In our experiments, \(Z\) is obtained by normalizing demonstration actions with the intended key and then comparing alternative decoders.  This isolates action-space semantics from policy quality: the same normalized trajectory is decoded under different executable specifications.

We define the mean displacement certificate
\begin{equation}
\Gamma(k_c,k_w;Z)
=\frac{1}{|Z|}\sum_{z\in Z}
\left\|\mathcal U_{k_w}(z)-\mathcal U_{k_c}(z)\right\|_2,
\label{eq:gamma}
\end{equation}
and the tail certificate
\begin{equation}
\Gamma_\tau(k_c,k_w;Z)
=\frac{1}{|Z|}\sum_{z\in Z}
\mathbf 1\left\{
\left\|\mathcal U_{k_w}(z)-\mathcal U_{k_c}(z)\right\|_2\ge \tau
\right\}.
\label{eq:gamma-tail}
\end{equation}
The threshold \(\tau\) is not treated as a universal safety margin.  It is a reporting scale for the action space.  Physical replay and dose-response experiments test whether the certified displacement is sufficient to change task outcomes.  We report \(F_{.2}\) in the main tables because it lies in the observed drift range for plausible keys and separates the strongest substitutions without serving as a deployment decision threshold.

We also use an interpolation family to probe causal dose response:
\begin{equation}
q_{01}^{(\alpha)}=(1-\alpha)q_{01}^{(c)}+\alpha q_{01}^{(w)},
\qquad
q_{99}^{(\alpha)}=(1-\alpha)q_{99}^{(c)}+\alpha q_{99}^{(w)},
\label{eq:alpha-interp}
\end{equation}
with \(\alpha\in[0,1]\).  At \(\alpha=0\), the executable specification is the intended one.  At \(\alpha=1\), it is the substituted metadata.  Intermediate values keep the normalized action sequence fixed while moving only the action-space semantics.  This is a controlled intervention on the executable policy, not a retraining or attack procedure.

\subsection{ExecSpec verifier}

The certificate suggests a minimal verifier.  Given an intended metadata record, a candidate deployment metadata record, and a calibration action set, the verifier computes Eq.~\eqref{eq:gamma}, Eq.~\eqref{eq:gamma-tail}, axis-wise displacement, mask mismatch, and gripper-dimension displacement.  The verifier flags large displacement before any model inference or robot rollout.

\paragraph{ExecSpec metadata verifier.}
Given intended metadata \(k_c\), candidate metadata \(k_w\), calibration physical actions \(A\), and reporting thresholds \(T\), the verifier performs five steps: (i) normalize \(A\) under \(k_c\) to obtain \(Z\); (ii) decode \(Z\) under both \(k_c\) and \(k_w\); (iii) compute per-sample displacement \(d(z)=\|\mathcal U_{k_w}(z)-\mathcal U_{k_c}(z)\|_2\); (iv) report \(\Gamma=\mathrm{mean}(d)\), quantiles of \(d\), and \(\Gamma_\tau\) for \(\tau\in T\); and (v) report mask mismatch, gripper displacement, and dominant action dimensions.  The output is an executable-specification certificate.

The verifier checks a boundary condition that is often absent from VLA model cards: whether the action-space semantics used at deployment match the semantics assumed by the checkpoint.  A high certificate does not prove that every rollout will fail, because task success also depends on state, dynamics, and recovery.  Conversely, a low certificate does not certify full deployment safety.  The certificate measures a necessary piece of the executable specification: how much the decoded action trajectory changes before the controller sees it.

\subsection{Connection to safety margins}

The certificate becomes safety-relevant when the metadata-induced displacement crosses a task's local action margin.  Let \(a_t^\star=\mathcal U_{k_c}(z_t)\) be an action that remains in the task-success basin at state \(s_t\), and let \(\mathcal A_{\mathrm{fail}}(s_t)\) denote actions that move the system into an unrecoverable or task-failing region.  Define the local margin
\begin{equation}
  m(s_t,a_t^\star)=
  \inf_{a\in\mathcal A_{\mathrm{fail}}(s_t)}\|a-a_t^\star\|_2.
\end{equation}
If \(\|\Delta(z_t;k_c,k_w)\|_2>m(s_t,a_t^\star)\), then the metadata substitution can move a successful normalized action outside the local success basin.  We do not assume access to \(m\) in deployment.  Instead, replay experiments provide witnesses: the same normalized action sequence is successful under \(k_c\) and fails under \(k_w\).  Such a witness proves that action-space metadata alone is sufficient to flip the physical outcome for that trajectory.

This safety framing is deliberately narrower than adversarial robustness.  No adversary changes the image, prompt, weights, or normalized action sequence.  The hazard arises because the executable policy is underspecified.  The verifier closes that gap by making action metadata an auditable component of the policy rather than an implicit configuration detail.

\section{Experiments}
\label{sec:experiments}

The experiments test whether the executable-specification hazard in Sec.~\ref{sec:certificate} is only a formal possibility or flips replayed physical outcomes under plausible metadata substitutions.  The intervention is deliberately narrow: the normalized action sequence, initial state, task, and simulator are fixed.  Only the metadata used to decode normalized actions is changed.  This isolates action-space semantics from model quality, prompt choice, visual perturbations, and adversarial optimization.

\subsection{Protocol}

For each LIBERO suite \cite{liu2023libero}, demonstrations are converted into normalized trajectories using the intended metadata key.  The same normalized trajectory is then decoded under either the intended key or a substituted key and replayed from the same recorded initial state.  A VLA-emitted normalized action would pass through the same metadata-selected unnormalizer; demonstration-derived trajectories are used to remove model-quality, sampling, perception, and prompt confounds from the executable-semantics test.  The experiment therefore asks a counterfactual identity question: if the internal action sequence is fixed, can changing only the metadata-selected executable map change the physical outcome?

A demonstration is counted as replay-valid if the intended, correct-decoded trajectory succeeds.  All replay success rates below are reported on this replay-valid set.  This choice avoids attributing failures of simulator replay itself to metadata mismatch.  As a sanity check, raw HDF5 replay succeeds on 27 of 30 LIBERO-Goal demonstrations, while correct-decoded replay succeeds on 28 of 30, with one mismatch; the paper therefore reports metadata substitutions on correct-decoded replay-valid demonstrations.  Section~\ref{sec:extended-evidence} spells out the replay-valid set construction, interpolation grid, calibration frame counts, and threshold convention.

The main plausible substitutions use sibling metadata records from the same VLA/LIBERO workflow: \textsc{spatial-v2}, \textsc{long-v2}, \textsc{mixed-all}, and \textsc{goal-B}.  These are not random corruptions.  They are keys that could plausibly be selected when a checkpoint, suite name, or dataset-statistics file is copied across related experiments.  We choose substituted keys by this workflow-level criterion, not by replay failure.  Synthetic permutations and shift-scale statistics are kept as upper-bound controls rather than as the deployment claim.

Calibration frames are sampled from the replay-valid demonstrations after applying the same frame filter used for replay.  LIBERO-Goal contributes 797 calibration frames and LIBERO-Spatial contributes 756.  Replay success uses the LIBERO task success predicate under the recorded initial state.

\subsection{Pre-rollout metadata certificates}

Table~\ref{tab:certificate-summary} reports the static ExecSpec certificates before replay.  On LIBERO-Goal, the intended key is \textsc{goal-B}.  All three plausible substitutions induce nonzero action-space displacement.  The strongest plausible substitution, \textsc{long-v2}, has mean drift 0.199 over six non-gripper action dimensions, p95 drift 0.275, p99 drift 0.318, and 43.5\% of calibration frames above drift 0.20.  On LIBERO-Spatial, the same substituted key yields mean drift 0.172 and 31.2\% of calibration frames above drift 0.20.  The certificate is computed without model inference and before replay.  Section~\ref{sec:expanded-certificates} decomposes these summaries into per-suite quantiles and per-dimension drift tables.

\begin{table*}[t]
\centering
\scriptsize
\setlength{\tabcolsep}{3pt}
\caption{Pre-rollout ExecSpec certificates on calibration frames.  Drift is over six normalized non-gripper dimensions.  \(F_{.2}\) is a reporting fraction at drift 0.20, not a deployment threshold.  Axes report top dimension:mean-absolute-drift; all rows have matching masks and zero gripper-dimension drift.}
\label{tab:certificate-summary}
\begin{tabular}{llrrrrrl}
\toprule
Suite & Key & Cal. $n$ & Mean & p95 & p99 & $F_{.2}$ & Top axes \\
\midrule
Goal & \textsc{spatial-v2} & 797 & 0.152 & 0.233 & 0.252 & 0.089 & $5$:.107,$0$:.067,$4$:.060 \\
Goal & \textsc{long-v2}   & 797 & 0.199 & 0.275 & 0.318 & 0.435 & $2$:.110,$5$:.110,$0$:.084 \\
Goal & \textsc{mixed-all} & 797 & 0.132 & 0.171 & 0.198 & 0.003 & $0$:.085,$5$:.076,$1$:.044 \\
Spatial & \textsc{goal-B}   & 756 & 0.178 & 0.223 & 0.263 & 0.135 & $5$:.139,$0$:.063,$4$:.063 \\
Spatial & \textsc{long-v2}  & 756 & 0.172 & 0.241 & 0.276 & 0.312 & $2$:.111,$1$:.070,$0$:.062 \\
Spatial & \textsc{mixed-all}& 756 & 0.113 & 0.160 & 0.209 & 0.017 & $5$:.074,$1$:.067,$4$:.028 \\
\bottomrule
\end{tabular}
\end{table*}

The magnitude is deployment-relevant because robot manipulation policies often operate close to contact and grasp margins.  A drift of this scale is not a formatting difference in a configuration file; it is a different physical action sent to the controller.  The replay experiments below test whether this certified displacement is sufficient to flip replayed task success.

\subsection{Full-substitution replay witnesses}

Table~\ref{tab:replay-witnesses} separates the demonstration-level replay outcome from the frame-level certificate.  Replay drift columns summarize frames from replay-valid intervention trajectories at full substitution \((\alpha=1)\), while success is measured by the LIBERO task predicate on replay-valid demonstrations.  We report counts and Wilson 95\% confidence intervals because the replay-valid sets are finite; the claim is witness-based rather than a deployment-prevalence estimate.

\begin{table*}[t]
\centering
\scriptsize
\setlength{\tabcolsep}{4pt}
\caption{Full-substitution replay witnesses.  Normalized trajectories and initial states are fixed; only the metadata-selected unnormalizer changes.  Wilson intervals are for the demonstration-level success rate.}
\label{tab:replay-witnesses}
\begin{tabular}{llrrrrl}
\toprule
Suite & Key & Replay $n$ & Drift mean & $F_{.2}$ & Success & Rate [95\% CI] \\
\midrule
Goal & \textsc{spatial-v2} & 28 & 0.151 & 0.083 & 8/28 & 0.286 [0.153, 0.471] \\
Goal & \textsc{long-v2}   & 28 & 0.202 & 0.472 & 2/28 & 0.071 [0.020, 0.226] \\
Goal & \textsc{mixed-all} & 28 & 0.132 & 0.004 & 5/28 & 0.179 [0.079, 0.356] \\
Spatial & \textsc{goal-B}   & 26 & 0.179 & 0.146 & 4/26 & 0.154 [0.062, 0.335] \\
Spatial & \textsc{long-v2}  & 26 & 0.173 & 0.345 & 0/26 & 0.000 [0.000, 0.129] \\
Spatial & \textsc{mixed-all}& 26 & 0.114 & 0.016 & 3/26 & 0.115 [0.040, 0.290] \\
\bottomrule
\end{tabular}
\end{table*}

\subsection{Object and Long replay witnesses}

We next test whether the endpoint replay witness extends beyond Goal and Spatial.  LIBERO-Object uses object-suite action statistics computed from Object actions as the intended key; because no separate released ECoT Object checkpoint statistics are used here, Object is treated as a suite-level executable-specification test rather than a checkpoint-specific claim.  Correct-decoded replay is valid for 28 of 30 demonstrations.  Under all four plausible sibling substitutions, replay success drops from 28/28 to 0/28.  LIBERO-Long uses \textsc{long-v2} as the intended key; correct-decoded replay is valid for 23 of 30 demonstrations.  Under full substitution, \textsc{goal-B} and \textsc{spatial-v2} both reduce success from 23/23 to 0/23, while \textsc{mixed-all} reduces success to 1/23.  Table~\ref{tab:object-long-replay} reports the full-substitution results.  These additional witnesses strengthen the executable-non-equivalence claim across all four LIBERO suite families tested here.  They remain replay witnesses, not hardware-failure prevalence estimates.

\begin{table*}[t]
\centering
\scriptsize
\setlength{\tabcolsep}{4pt}
\caption{Object and Long full-substitution replay witnesses.  Object uses object-suite action statistics as the intended key; Long uses \textsc{long-v2}.  Wilson intervals are for the demonstration-level success rate on replay-valid demonstrations.}
\label{tab:object-long-replay}
\begin{tabular}{llrrrrl}
\toprule
Suite & Key & Replay $n$ & Drift mean & $F_{.2}$ & Success & Rate [95\% CI] \\
\midrule
Object & \textsc{goal-B} & 28 & 0.264 & 0.896 & 0/28 & 0.000 [0.000, 0.121] \\
Object & \textsc{spatial-v2} & 28 & 0.158 & 0.170 & 0/28 & 0.000 [0.000, 0.121] \\
Object & \textsc{long-v2} & 28 & 0.179 & 0.269 & 0/28 & 0.000 [0.000, 0.121] \\
Object & \textsc{mixed-all} & 28 & 0.147 & 0.103 & 0/28 & 0.000 [0.000, 0.121] \\
Long & \textsc{goal-B} & 23 & 0.197 & 0.418 & 0/23 & 0.000 [0.000, 0.143] \\
Long & \textsc{spatial-v2} & 23 & 0.159 & 0.132 & 0/23 & 0.000 [0.000, 0.143] \\
Long & \textsc{mixed-all} & 23 & 0.136 & 0.025 & 1/23 & 0.043 [0.008, 0.210] \\
\bottomrule
\end{tabular}
\end{table*}

\subsection{Replay dose response}

Table~\ref{tab:dose-response} evaluates the interpolation in Eq.~\eqref{eq:alpha-interp} on both suites.  At \(\alpha=0\), all replay-valid demonstrations remain successful.  As metadata moves away from the intended key, success drops sharply.  For LIBERO-Goal with \textsc{long-v2}, success falls from 28/28 at \(\alpha=0\) to 17/28 at \(\alpha=0.25\), 7/28 at \(\alpha=0.50\), 0/28 at \(\alpha=0.75\), and 2/28 at \(\alpha=1.00\).  For LIBERO-Spatial with the same substituted key, success falls from 26/26 to 9/26, 1/26, 0/26, and 0/26.  The interpolation is monotone in metadata space, but task success is a thresholded dynamical predicate; individual trajectories can fail and recover across interpolation values depending on contact geometry and controller response.  The trend is dose-dependent but not claimed to be strictly monotone.

\begin{table*}[t]
\centering
\scriptsize
\setlength{\tabcolsep}{3pt}
\caption{Replay success under metadata interpolation.  Normalized trajectories and initial states are fixed; only action unnormalization metadata changes.  Endpoint drift columns summarize replay trajectories at \(\alpha=1\).}
\label{tab:dose-response}
\begin{tabular}{llrrrrrrrr}
\toprule
Suite & Key & $\alpha=0$ & $0.25$ & $0.50$ & $0.75$ & $1.00$ & Drift@1 & p95@1 & $F_{.2}$@1 \\
\midrule
Goal & \textsc{spatial-v2} & 28/28 & 22/28 & 18/28 & 13/28 & 8/28 & 0.151 & 0.191 & 0.083 \\
Goal & \textsc{long-v2}   & 28/28 & 17/28 &  7/28 &  0/28 & 2/28 & 0.202 & 0.275 & 0.472 \\
Goal & \textsc{mixed-all} & 28/28 & 20/28 & 15/28 &  9/28 & 5/28 & 0.132 & 0.171 & 0.004 \\
\midrule
Spatial & \textsc{goal-B}   & 26/26 & 20/26 & 14/26 &  7/26 & 4/26 & 0.179 & 0.221 & 0.146 \\
Spatial & \textsc{long-v2}  & 26/26 &  9/26 &  1/26 &  0/26 & 0/26 & 0.173 & 0.239 & 0.345 \\
Spatial & \textsc{mixed-all}& 26/26 & 14/26 &  7/26 &  4/26 & 3/26 & 0.114 & 0.157 & 0.016 \\
\bottomrule
\end{tabular}
\end{table*}

The dose response gives a direct witness for executable non-equivalence.  The same normalized action trajectories that succeed under the intended metadata fail under plausible sibling metadata.  We report success counts rather than only percentages because the replay-valid sets are small.  No image perturbation, prompt edit, model weight change, or policy sampling change is involved.

\subsection{Four-suite replay summary}

Across Goal, Spatial, Object, and Long, metadata substitution changes the decoded action law while the normalized trajectory is fixed, and replay-valid trajectories frequently leave the success basin.  Goal and Spatial provide the dose-response and attribution setting.  Object and Long provide additional full-substitution witnesses across the remaining LIBERO suite families.  The endpoint effect is established; the next step is failure attribution.  The next section decomposes the certificate and rules out replay artifacts, no-op filtering, and gripper-convention mismatch.

\section{Extended Evidence and Failure Attribution}
\label{sec:extended-evidence}

The replay tables above establish the endpoint phenomenon: under fixed normalized trajectories and fixed initial states, plausible sibling metadata keys can change physical replay outcomes.  This section gives the failure-attribution chain.  It separates replay validity, certificate decomposition, and alternative explanations such as no-op filtering or gripper convention.

The checks below give a consistent account.  Correct-decoded replay is used to define replay-valid demonstrations; the certificate decomposes into non-gripper action displacement under shared masks; larger drift is associated with higher failure rates within the protocol; and identity-key, no-op, raw-vs-correct, and synthetic controls behave as expected.  These checks do not turn ExecSpec into a task-success predictor.  They support the narrower claim that metadata substitution changes the executable action semantics before rollout, and that this semantic shift can be sufficient to leave the success basin.

\subsection{Replay-valid evaluation set}

The first attribution issue is replay validity.  For each suite, we start from recorded demonstrations and construct normalized trajectories under the intended metadata key.  These normalized trajectories are then decoded twice: once under the intended key and once under a substituted key.  A demonstration enters the replay-valid set only when the intended, correct-decoded trajectory succeeds from the recorded initial state.  This filtering does not favor the metadata substitution.  It removes demonstrations for which simulator replay itself is already unreliable, so that subsequent failures are attributable to changing the executable action semantics.

LIBERO-Goal uses 30 demonstrations before replay filtering.  Correct-decoded replay succeeds on 28 of them.  LIBERO-Spatial also starts from 30 demonstrations and has 26 correct-decoded replay-valid demonstrations.  LIBERO-Object has 28 replay-valid demonstrations out of 30.  LIBERO-Long has 23 replay-valid demonstrations out of 30.  All success rates in the metadata-substitution tables are reported on these replay-valid sets.

\subsection{Metadata keys and interventions}

The plausible substitutions are sibling metadata records from the same VLA/LIBERO workflow.  For LIBERO-Goal, the intended key is \textsc{goal-B}; the substituted keys are \textsc{spatial-v2}, \textsc{long-v2}, and \textsc{mixed-all}.  For LIBERO-Spatial, the intended key is \textsc{spatial-v2}; the substituted keys are \textsc{goal-B}, \textsc{long-v2}, and \textsc{mixed-all}.  The substitutions are therefore not random corruptions.  They model a deployment error in which a nearby dataset, suite, or mixed-statistics record is selected as the action unnormalizer.  We choose them by this workflow-level criterion, not by replay failure.

The interpolation experiment changes only the unnormalization statistics:
\begin{equation}
\begin{aligned}
q_{01}^{(\alpha)}&=(1-\alpha)q_{01}^{(c)}+\alpha q_{01}^{(w)},\\
q_{99}^{(\alpha)}&=(1-\alpha)q_{99}^{(c)}+\alpha q_{99}^{(w)}.
\end{aligned}
\end{equation}
The normalized trajectory is fixed for all \(\alpha\).  This makes the dose-response experiment a controlled intervention on the executable specification, not a new policy rollout or an adversarial optimization procedure.

\subsection{Calibration actions and thresholds}

The certificate is computed on calibration frames obtained by normalizing demonstration actions under the intended key and comparing alternative decoders.  LIBERO-Goal has 797 calibration frames after the protocol's frame sampling.  LIBERO-Spatial has 756 calibration frames.  The main threshold \(\tau=0.20\) is used only as a reporting scale in the normalized non-gripper action coordinates.  It is not a universal safety threshold or a deployment decision rule.

The verifier reports mean L2 displacement, quantiles of L2 displacement, threshold fractions, axis-wise absolute displacement, mask mismatch, and gripper displacement.  Replay then tests whether the displacement changes task outcome.  The certificate and replay play different roles: the certificate identifies executable non-equivalence before rollout, while replay supplies physical witnesses that the change can leave the task success basin.

\subsection{Certificate decomposition}
\label{sec:expanded-certificates}

The second attribution issue is whether the certificate measures the same semantic change that drives replay failure.  The LIBERO-Goal certificate decomposes the compact main-text summary into quantiles, threshold fractions, and action-axis attribution.  The strongest plausible substitution, \textsc{long-v2}, has mean drift 0.199, p95 drift 0.275, p99 drift 0.318, and 43.5\% of calibration frames above drift 0.20.  The top non-gripper dimensions by mean absolute drift are dimensions 2, 5, and 0.  The gripper dimension has zero mean drift.

\begin{table*}[t]
\centering
\scriptsize
\setlength{\tabcolsep}{3pt}
\caption{Expanded LIBERO-Goal metadata certificate.  Drift is over normalized non-gripper dimensions.  Axis columns report mean absolute displacement per action dimension; dimension 6 is the gripper dimension and has zero drift for all plausible substitutions.}
\label{tab:app-goal-certificate}
\begin{tabular}{lrrrrrrrrrrrrrrrr}
\toprule
Key & $n$ & Mean & p50 & p90 & p95 & p99 & $F_{.1}$ & $F_{.2}$ & $F_{.3}$ & d0 & d1 & d2 & d3 & d4 & d5 & d6 \\
\midrule
\textsc{spatial-v2} & 797 & 0.152 & 0.145 & 0.191 & 0.233 & 0.252 & 0.996 & 0.089 & 0.000 & 0.067 & 0.037 & 0.001 & 0.023 & 0.060 & 0.107 & 0.000 \\
\textsc{long-v2} & 797 & 0.199 & 0.191 & 0.252 & 0.275 & 0.318 & 1.000 & 0.435 & 0.026 & 0.084 & 0.044 & 0.110 & 0.010 & 0.038 & 0.110 & 0.000 \\
\textsc{mixed-all} & 797 & 0.132 & 0.130 & 0.160 & 0.171 & 0.198 & 0.911 & 0.003 & 0.000 & 0.085 & 0.044 & 0.000 & 0.014 & 0.030 & 0.076 & 0.000 \\
\bottomrule
\end{tabular}
\end{table*}

LIBERO-Spatial shows the same certificate structure.  The \textsc{long-v2} substitution again produces non-gripper displacement with zero gripper drift, and it drives replay success to 0/26 at full substitution in the main-text table.

\begin{table*}[t]
\centering
\scriptsize
\setlength{\tabcolsep}{3pt}
\caption{Expanded LIBERO-Spatial metadata certificate.  Drift is over normalized non-gripper dimensions.  Axis columns report mean absolute displacement per action dimension; dimension 6 is the gripper dimension and has zero drift for all plausible substitutions.}
\label{tab:app-spatial-certificate}
\begin{tabular}{lrrrrrrrrrrrrrrrr}
\toprule
Key & $n$ & Mean & p50 & p90 & p95 & p99 & $F_{.1}$ & $F_{.2}$ & $F_{.3}$ & d0 & d1 & d2 & d3 & d4 & d5 & d6 \\
\midrule
\textsc{goal-B} & 756 & 0.178 & 0.179 & 0.206 & 0.223 & 0.263 & 0.997 & 0.135 & 0.000 & 0.063 & 0.033 & 0.001 & 0.034 & 0.063 & 0.139 & 0.000 \\
\textsc{long-v2} & 756 & 0.172 & 0.170 & 0.229 & 0.241 & 0.276 & 0.956 & 0.312 & 0.001 & 0.062 & 0.070 & 0.111 & 0.023 & 0.019 & 0.053 & 0.000 \\
\textsc{mixed-all} & 756 & 0.113 & 0.110 & 0.146 & 0.160 & 0.209 & 0.640 & 0.017 & 0.000 & 0.017 & 0.067 & 0.001 & 0.013 & 0.028 & 0.074 & 0.000 \\
\bottomrule
\end{tabular}
\end{table*}

\subsection{Object and Long certificates and replay witnesses}
\label{sec:object-long-certificates}

The main text reports Goal and Spatial in detail.  We also run the same static verifier and full-substitution replay protocol on LIBERO-Object and LIBERO-Long.  Object uses object-suite action statistics as the intended metadata; because no separate released ECoT Object checkpoint statistics are used here, this row is a suite-level executable-specification test rather than a checkpoint-specific claim.  Long uses \textsc{long-v2}.  Table~\ref{tab:app-object-long-certificate} reports the static certificates.  Table~\ref{tab:app-object-long-replay} reports replay success on replay-valid demonstrations.  These rows extend the endpoint witness across all four LIBERO suite families tested here while preserving the same claim boundary: the result is replay evidence under fixed normalized trajectories, not a hardware failure-rate estimate.

\begin{table*}[t]
\centering
\scriptsize
\setlength{\tabcolsep}{3pt}
\caption{Object and Long metadata certificates.  Drift is over normalized non-gripper dimensions.  All rows have matching masks and zero gripper-dimension drift.}
\label{tab:app-object-long-certificate}
\begin{tabular}{llrrrrrrl}
\toprule
Suite & Key & $n$ & Mean & p95 & p99 & $F_{.2}$ & $F_{.3}$ & Top axes \\
\midrule
Object & \textsc{goal-B} & 913 & 0.264 & 0.372 & 0.423 & 0.907 & 0.210 & $5$:.164,$0$:.157,$1$:.101 \\
Object & \textsc{spatial-v2} & 913 & 0.158 & 0.230 & 0.257 & 0.168 & 0.000 & $1$:.114,$0$:.083,$4$:.041 \\
Object & \textsc{long-v2} & 913 & 0.178 & 0.252 & 0.302 & 0.265 & 0.013 & $2$:.100,$5$:.089,$0$:.084 \\
Object & \textsc{mixed-all} & 913 & 0.147 & 0.225 & 0.287 & 0.100 & 0.005 & $5$:.107,$0$:.064,$1$:.057 \\
Long & \textsc{goal-B} & 1697 & 0.196 & 0.271 & 0.308 & 0.413 & 0.012 & $2$:.112,$5$:.100,$0$:.088 \\
Long & \textsc{spatial-v2} & 1697 & 0.159 & 0.217 & 0.245 & 0.115 & 0.000 & $2$:.114,$1$:.070,$0$:.046 \\
Long & \textsc{mixed-all} & 1697 & 0.136 & 0.190 & 0.208 & 0.019 & 0.000 & $2$:.112,$0$:.049,$5$:.030 \\
\bottomrule
\end{tabular}
\end{table*}

\begin{table*}[t]
\centering
\scriptsize
\setlength{\tabcolsep}{4pt}
\caption{Object and Long full-substitution replay witnesses.  Wilson intervals are for the demonstration-level success rate on replay-valid demonstrations.}
\label{tab:app-object-long-replay}
\begin{tabular}{llrrrrl}
\toprule
Suite & Key & Replay $n$ & Drift mean & $F_{.2}$ & Success & Rate [95\% CI] \\
\midrule
Object & \textsc{goal-B} & 28 & 0.264 & 0.896 & 0/28 & 0.000 [0.000, 0.121] \\
Object & \textsc{spatial-v2} & 28 & 0.158 & 0.170 & 0/28 & 0.000 [0.000, 0.121] \\
Object & \textsc{long-v2} & 28 & 0.179 & 0.269 & 0/28 & 0.000 [0.000, 0.121] \\
Object & \textsc{mixed-all} & 28 & 0.147 & 0.103 & 0/28 & 0.000 [0.000, 0.121] \\
Long & \textsc{goal-B} & 23 & 0.197 & 0.418 & 0/23 & 0.000 [0.000, 0.143] \\
Long & \textsc{spatial-v2} & 23 & 0.159 & 0.132 & 0/23 & 0.000 [0.000, 0.143] \\
Long & \textsc{mixed-all} & 23 & 0.136 & 0.025 & 1/23 & 0.043 [0.008, 0.210] \\
\bottomrule
\end{tabular}
\end{table*}

\subsection{OpenVLA-style unnormalizer interface sanity}
\label{sec:openvla-interface-sanity}

The replay experiments use demonstration-derived normalized trajectories to isolate executable semantics.  A separate interface-rule sanity check verifies that the same metadata-selected unnormalization rule applies at an OpenVLA-style action interface.  This is not a learned-model forward pass; it checks the controller-facing decoding rule that a model-emitted normalized action would use.  For an OpenVLA-style action normalizer, the controller-facing action is computed as
\begin{equation}
  a_i=\begin{cases}
  z_i(q_{99,i}-q_{01,i})+q_{01,i}, & m_i=1,\\
  z_i, & m_i=0.
  \end{cases}
\end{equation}
Using the same normalized action vectors and swapping only the metadata key changes the decoded action by the ExecSpec affine transform to float precision.  Table~\ref{tab:app-interface-sanity} reports the resulting displacement under three candidate keys.  This check does not evaluate model accuracy; it verifies that model-emitted normalized actions and demonstration-derived normalized actions share the same metadata-selected executable interface.

\begin{table}[t]
\centering
\small
\caption{OpenVLA-style metadata-selected unnormalizer interface sanity.  The same normalized actions are decoded under different metadata keys.  Formula error is the maximum absolute difference between direct decoding and the closed-form ExecSpec displacement.}
\label{tab:app-interface-sanity}
\begin{tabular}{lrrr}
\toprule
Candidate key & Mean L2 & Formula error & Mask mismatch \\
\midrule
\textsc{spatial-v2} & 0.185 & $5.96\times 10^{-8}$ & no \\
\textsc{long-v2} & 0.212 & $2.98\times 10^{-8}$ & no \\
\textsc{mixed-all} & 0.137 & $2.98\times 10^{-8}$ & no \\
\bottomrule
\end{tabular}
\end{table}

\subsection{Drift--failure association}
\label{sec:replay-details}

The third attribution issue is whether certificate magnitude is connected to replay outcome within the controlled protocol.  The main text reports endpoint success and the interpolation dose response.  Here we unpack the association between drift and failure on LIBERO-Goal.  Across 336 replay-valid intervention points, the correlation between mean drift and binary failure is 0.472.  The correlation is 0.495 for \textsc{long-v2}, 0.438 for \textsc{spatial-v2}, and 0.413 for \textsc{mixed-all}.

\begin{table}[t]
\centering
\small
\caption{LIBERO-Goal failure rate by drift bin across replay-valid intervention points.  The certificate is not a classifier, but higher drift bins correspond to much higher replay failure rates in this protocol.}
\label{tab:app-drift-bins}
\begin{tabular}{lrr}
\toprule
Drift bin & $n$ & Failure rate \\
\midrule
$[0,0.05)$ & 70 & 0.286 \\
$[0.05,0.10)$ & 103 & 0.485 \\
$[0.10,0.15)$ & 112 & 0.732 \\
$[0.15,0.20)$ & 37 & 0.946 \\
$[0.20,0.25)$ & 14 & 0.929 \\
\bottomrule
\end{tabular}
\end{table}

This table should not be read as a universal calibration curve.  It is a within-protocol diagnostic: once the normalized trajectory and initial state are fixed, larger action-space displacement is associated with a higher probability that the replay leaves the success basin.  The result supports using ExecSpec as a pre-rollout interface screen, not as a replacement for task-specific rollout evaluation.

\subsection{Alternative explanations and upper-bound controls}
\label{sec:alternative-controls}

We finally test simpler explanations for the replay failures.  The identity-key control substitutes the intended metadata key with itself.  It has zero drift at every interpolation value and preserves success on all 28 LIBERO-Goal replay-valid demonstrations.  This verifies that the replay pipeline does not create failures when the executable specification is unchanged.

Disabling no-op filtering leaves the main LIBERO-Goal endpoint success rates unchanged at full substitution: \textsc{spatial-v2} remains 8/28, \textsc{long-v2} remains 2/28, and \textsc{mixed-all} remains 5/28.  The metadata-mismatch result is therefore not an artifact of a no-op filtering rule.

Raw HDF5 replay succeeds on 27 of 30 LIBERO-Goal demonstrations.  Correct-decoded replay succeeds on 28 of 30, with one mismatch between the two replay modes.  The main analysis uses correct-decoded replay-valid demonstrations.  This choice isolates the intervention studied in the paper: changing the metadata used to decode a fixed normalized trajectory.

The gripper convention also does not explain the failures.  All plausible substitutions have the same normalization mask and zero gripper-dimension drift.  On LIBERO-Goal, the dominant \textsc{long-v2} drift dimensions are non-gripper dimensions 2, 5, and 0, with mean absolute displacements 0.110, 0.110, and 0.084.  On LIBERO-Spatial, \textsc{long-v2} again has zero gripper drift, with dominant non-gripper dimensions 2, 1, and 0.  The replay failures are therefore caused by hidden action-space semantics, not a gripper sign flip.

Synthetic metadata corruptions are not used as the deployment claim.  They calibrate the upper end of the certificate.  A permuted metadata record has mean drift 0.305, p95 drift 0.597, and 1/28 success at \(\alpha=1\).  A deterministic shift-scale corruption has mean drift 0.607, p95 drift 0.879, and 0/28 success.  These controls show that the verifier responds to large action-space distortions as expected, while the paper's main claim remains the plausible-key result.

\begin{table}[t]
\centering
\small
\caption{Synthetic upper-bound controls on LIBERO-Goal.  These rows are controls, not the deployment claim.}
\label{tab:app-synthetic}
\begin{tabular}{lrrrr}
\toprule
Control & Mean & p95 & $F_{.2}$ & Success \\
\midrule
Synthetic permuted & 0.305 & 0.597 & 0.644 & 1/28 \\
Synthetic shift-scale & 0.607 & 0.879 & 1.000 & 0/28 \\
\bottomrule
\end{tabular}
\end{table}

\paragraph{Reproducibility note.}
All reported tables use the corrected LIBERO task ordering specified by the replay protocol.  Runs with inconsistent task ordering are excluded from the manuscript tables.  The verifier is deterministic given the intended metadata record, candidate metadata record, calibration actions, replay-valid demonstrations, and interpolation grid \(\alpha\in\{0,0.25,0.5,0.75,1.0\}\).  The replay witness depends on simulator state reset and demonstration initial states, which is why endpoint success rates are reported on replay-valid demonstrations rather than on the full unfiltered set.

\section{Deployment Implications and Limitations}
\label{sec:discussion}

\paragraph{From certificate to manifest.}
The attribution analysis changes the deployment lesson.  The relevant safety object is not the checkpoint alone, but the executable policy: weights plus the controller-facing semantics that turn model outputs into robot actions.  A deployment that changes the action unnormalization key is therefore not a minor configuration variant of the same policy.  It is a different executable policy, even if the prompt, benchmark suite, normalized action sequence, and neural weights are unchanged.

This motivates a manifest discipline for VLA deployment.  At minimum, the manifest should record the checkpoint hash, action unnormalization key, statistics-record hash, normalization mask, action dimensionality, gripper convention, absolute or relative action convention, control frequency, and action hold duration.  For tokenized or chunked policies, it should also record the chunk-execution convention, image preprocessing, camera orientation, proprioceptive fields, and token-to-action decoding convention.  ExecSpec certifies one field in this manifest: quantile-style action unnormalization metadata.  The general rule is broader: if two deployments differ in a controller-facing semantic field, they should not be treated as the same robot policy merely because their model weights match.

\paragraph{What ExecSpec certifies.}
ExecSpec is a pre-rollout interface check, not a proof of task safety.  A large certificate identifies a metadata substitution that changes physical action semantics before the robot moves.  A small certificate only says that this particular interface mismatch was not detected at the chosen scale; it does not certify success, robustness, or safe contact behavior.  Perception, dynamics, recovery, controller saturation, and the learned policy itself remain outside the certificate.  The useful operational interpretation is asymmetric: failing the check is strong evidence that the deployment is not the intended executable policy, while passing it only clears one necessary deployment condition.

\paragraph{Why replay is the right isolation test.}
The experiments use replayed normalized trajectories rather than learned-policy closed-loop inference.  This is deliberate.  Closed-loop VLA failures can be caused by visual preprocessing, prompt format, sampling, model error, compounding drift, simulator mismatch, or recovery behavior.  Replay holds those factors fixed and isolates the executable-semantics question.  If an expert trajectory succeeds under the intended metadata and fails under substituted metadata, then action-space semantics alone are sufficient to flip the physical outcome for that trajectory.  This does not replace policy rollout evaluation; it identifies a lower-level safety condition that should be checked before policy quality is evaluated.

\paragraph{Scope of the empirical claim.}
The empirical claim is limited to LIBERO-Goal, LIBERO-Spatial, LIBERO-Object, and LIBERO-Long replay under the metadata records tested here.  The paper does not claim that every VLA checkpoint is vulnerable to every metadata substitution, that normalization mismatch is the dominant failure mode in robot learning, or that the measured replay failures predict hardware failure rates.  The supported claim is more specific: plausible sibling metadata keys in a VLA/LIBERO workflow induce measurable action-space transforms, and those transforms can be sufficient to flip replayed task success while normalized actions and initial states are fixed.  The formal result applies to quantile-style unnormalizers in general, but physical consequence depends on task margins, controller conventions, and dynamics.

\paragraph{Thresholds are reporting scales, not safety margins.}
The threshold \(\tau=0.20\) used in the tables is a reporting scale, not a universal robot safety margin.  Different embodiments, controllers, and action units require different scales.  The more informative quantities are the full drift distribution, axis attribution, replay outcome, and dose-response curve.  In our experiments, larger certificates align with higher replay failure rates, but the certificate is not a perfect predictor.  Future work should calibrate certificates against task-specific margins, contact-sensitive safety metrics, and hardware-specific action units.

\paragraph{Metadata beyond normalization.}
This paper focuses on action unnormalization because it gives a closed-form certificate and clean physical witnesses.  The executable-policy view is broader.  Image resizing, camera orientation, tokenizer special tokens, gripper sign conventions, control frequency, and action chunk execution can also be hidden parts of the policy.  Some of these components may admit similar static certificates; others may require simulator or hardware checks.  A complete VLA deployment manifest should cover controller-facing semantics beyond normalization statistics.

\paragraph{Hardware deployment.}
The experiments use LIBERO simulation replay.  This is appropriate for isolating action semantics, but it is not a substitute for real-robot deployment testing.  On hardware, metadata mismatch could interact with controller saturation, compliance, calibration error, emergency stops, and contact-rich recovery.  These interactions may reduce or amplify the observed effect.  The safety recommendation is therefore conservative: detect and block executable-specification mismatch before either simulation rollout or hardware execution.

\section{Conclusion}
\label{sec:conclusion}

A robot policy is not fully specified by its neural weights.  For VLA deployment, action-space metadata can determine the physical action law that reaches the controller.  ExecSpec makes this hidden interface explicit: for quantile-style normalizers, metadata mismatch induces a closed-form action displacement that can be checked before rollout.  The LIBERO replay witnesses show why this interface matters, while the limitations above bound the claim to replayed trajectories and metadata records tested here.

\end{document}